\documentclass[twoside]{aiml22}

\usepackage{aiml22macro}

\usepackage{graphicx}
\usepackage{amsmath}
\usepackage{amssymb}

\newcommand{\free}{\mathrm{free}}

\def\lastname{van Benthem, ten Cate and Koudijs}

\begin{document}

\begin{frontmatter}
  \title{Local Dependence and Guarding}
   \author{Johan van Benthem}
  \address{ILLC University of Amsterdam, Stanford University, Tsinghua University}
  \author{Balder ten Cate}\footnote{Supported by the European Union’s Horizon 2020 research and innovation programme under grant MSCA-101031081.}
  \address{ILLC University of Amsterdam}
  \author{Raoul Koudijs}
  \address{ILLC University of Amsterdam}

  \begin{abstract}
We study LFD, a base  logic of functional dependence introduced by
Baltag and van Benthem (2021) and  its connections with the guarded 
fragment GF of first-order logic. Like other logics of dependence, the 
semantics of LFD uses teams: sets of permissible variable assignments. 
What sets LFD apart is its ability to express local dependence between 
variables and  local dependence of statements on variables.
 
Known features of LFD include decidability,  explicit  axiomatization,
finite model property, and a bisimulation characterization. Others, including 
the complexity of  satisfiability, remained open so far. More generally, what 
has been lacking is a good understanding of what makes the LFD approach 
to dependence computationally well-behaved, and how it relates to other 
decidable logics. In particular, how do allowing variable dependencies and 
guarding quantifiers compare as logical devices?
 
We provide a new compositional translation from GF into LFD, and conversely, 
we translate LFD into GF in an `almost compositional' manner. Using these two  
translations, we transfer known results about GF to LFD in a uniform manner, yielding, 
e.g., tight complexity bounds for LFD satisfiability, as well as Craig interpolation. 
Conversely, e.g., the finite model property of LFD transfers to GF. Thus, local dependence 
and guarding turn out to be intricately entangled notions. 
  \end{abstract}

 \end{frontmatter}

\section{Introduction:  from guarding to dependence}

The Guarded Fragment GF of first-order logic, \cite{ModLangBoundedFrag}, is a well-known formalism for quantifying over tuples of objects that are locally \emph{guarded}  by predicates. Guarded existential quantification has the format $\exists \bar{y}(G(\bar{x}, \bar{y})\, \land\, \varphi(\bar{x}, \bar{y}))$, where $\bar{x}$, $\bar{y}$ are finite sets or tuples of variables, the  atom $G(\bar{x}, \bar{y})$ is the guard, and $\varphi(\bar{x}, \bar{y})$ is  a guarded formula having at most the free variables displayed. For further literature on GF and its extensions, see    \cite{Graedelguards}, \cite{GNbcs}.

Guarding quantifiers in truth conditions for  existing semantics is  a general  device for lowering  complexity of logical systems, \cite{ABBN}. But there are more such devices. In particular, \cite{ModLangBoundedFrag} connects guarding with  `generalized assignment models' for first-order logic [to be called `teams' henceforth in accordance with modern terminology]  where not all maps from variables to objects need to be present. These `gaps' in the full function space can be seen as modeling \emph{dependence}, a major topic in the current logical literature, \cite{Alice}, \cite{Vaa}. With  assignments missing, correlations  arise, as changing the value of one variable $x$ may only be feasible in the given  space by also changing that of another variable $y$. In contrast, in standard models for first-order logic all variables  take their values independently. The  origins of this semantics lie in   `cylindric relativized set algebra' CRS in algebraic logic where standard set algebras are relativized to one fixed relation over the domain, \cite{Nemeti86}. Unlike first-order logic, CRS logic is decidable.

That guarding can encode dependence was proved in \cite{ModLangBoundedFrag}.   First-order formulas in CRS  semantics translate compositionally into guarded formulas over standard first-order models by coding the available tuples of values for  finite tuples of relevant variables $\bar{v}$ as one new guard predicate $G(\bar{v})$. Conversely, \cite{Guardsbounds} gave an effective reduction for the satisfiability problem for GF formulas $\varphi$  to that for first-order formulas over team models, though the translation employed  is not compositional.\footnote{Related results and further general analysis  are found in \cite{marxtaming}.}Thus, dependence models can also encode guarding. The main topic of this paper is a further technical analysis of how far this analogy goes, now based on modern dependence logics rather than CRS.

CRS-style  semantics treats dependence implicitly by  giving up classical FO laws such as Commutativity $\exists x \exists y\, \varphi(x, y) \leftrightarrow \exists y \exists x\, \varphi(x, y)$ which express independence of the variables $x$ and $y$. In a next step,  \cite{Vaa}  introduced explicit syntactic atoms $D_Xy$ expressing functional dependence of  variable $y$ on the set of variables $X$. The  first generation of dependence logics over this richer language was second-order and non-classical, but many simpler fragments have been studied. A systematic use of GF-style guarding to lower the complexity of dependence logics is found in \cite{GrOt}. However, the present paper is concerned with a simpler base system for dependence logic.

A new base logic LFD for a first-order  language with explicit dependence atoms was recently proposed in \cite{LFD}  as a minimal way of reasoning about dependence on a classical base relying on  a modal local semantics. LFD is a dualized CRS logic [in a sense to be defined below] plus dependence atoms, which is decidable, axiomatizable, and allows for natural language extensions, e.g., to function symbols and  independence modalities. Many of these properties involve GF-style methods such as the use of type models and representation theorems. Thus, the question arises whether the connection between guarding and dependence extends to this new setting. But there are  obstacles here. The  first-order translation for LFD lands  in the guarded fragment, except for the  dependence atoms, cf. \cite{localdeps}.
In this paper, we offer the following new results. 

\begin{itemize}
    \item 
We give a compositional translation from GF into LFD without dependence atoms, considerably improving the known SAT reduction. 

\item 
We give an effective SAT reduction from LFD into GF, answering a question left open in the literature. 

\item We show how these results allow for transfer of important properties between LFD and GF. For instance, the finite model property for GF~\cite{Graedelguards}, can be derived from that for LFD, established in~\cite{Raoul}. Conversely, the decidability of LFD follows from that of GF. Moreover, using our translations, we solve the open problem of the  computational complexity of LFD.

\item Beyond  single examples, we also discuss the general transfer of system properties made possible by our translations, and discuss more general consequences for relating decidable fragments of first-order logic.
\end{itemize}

All missing proofs can be found in the appendix.

\section{Preliminaries}\label{sec:prel}


\subsection{The Logic of Functional Dependence LFD}

The Logic of Functional Dependencies (LFD) was introduced in~\cite{LFD}, to which we refer for all definitions and results stated in this section.

\begin{definition}(Syntax of LFD) ~
Fix a finite set of variables $V_{LFD}$, and a 
relational signature (i.e., a set of relation symbols with associated arity)
$\mathbf{S}$. The formulas of LFD are recursively generated by the following grammar.
\[\psi::= P(v_1, \ldots, v_n) \mid D_Vu \mid \psi\wedge\psi \mid \neg\psi \mid \mathbb{E}_V\psi\]
with $v_1, \ldots, v_n,u\in V_{LFD}$,  $V\subseteq V_{LFD}$, and 
$P\in\mathbf{S}$ an $n$-ary relation symbol.
\end{definition}

We read $D_Vu$ as `$u$ locally depends (only) on $V$' or   `$u$ is locally determined by $V$'. 
For notational convenience, we write $D_VU$ for the conjunction $\bigwedge_{u\in U}D_Vu$. We also write $D_{v}u$ as a shorthand for $D_{\{v\}}u$, and $\mathbb{E}_{v}\psi$ for $\mathbb{E}_{\{v\}}\psi$.

The semantics of LFD uses \emph{dependence models}, also known as `generalized assignment models'. These are pairs $\mathbb{M}=(M,A)$ where $M$ is a 
standard model over the relational signature $\mathbf{S}$,
and $A\subseteq dom(M)^{V_{LFD}}$ is a collection of admissible variable assignments, 
also known as a `team'.

\begin{definition}{(Semantics of LFD)}\\
Truth of a formula $\varphi$ in a dependence 
model $\mathbb{M}=(M,A)$ under an assignment $s\in A$ is defined as follows, where $s =_V t$ means that $s(v)=t(v)$ for all $v\in V$:

\vspace{-1ex}

\[\begin{array}{lll}
\mathbb{M},s\models P(v_1,\ldots,v_n) &\text{iff}      & s(v_1,\ldots,v_n)\in P^M\\

\vspace{1mm}

\mathbb{M},s\models \varphi\land\psi &\text{iff} & 
  \mathbb{M},s\models \varphi \text{ and } \mathbb{M},s\models \psi \\
  
  \vspace{1mm}
\mathbb{M},s\models \neg\varphi &\text{iff} & \text{not} \,
  \mathbb{M},s  \models \varphi \\
  
  \vspace{1mm}
\mathbb{M},s\models\mathbb{E}_V\varphi &\textrm{iff}& 
  \text{for some } t\in A \textrm{ , } s=_V t \textrm{ and } \mathbb{M},t\models\varphi\\
  
  \vspace{1mm}
\mathbb{M},s\models D_Vu &\textrm{iff}& \textrm{for all $t\in A$, $s=_V t$ implies $s=_u t$}
\end{array}\]

\end{definition}

LFD formulas come with a notion of \textit{free variable}, being a variable whose value can affect the truth of the formula, in the semantics to follow.

\begin{definition}{(Free Variables)}\\
The free variables of an LFD formula $\psi$ are defined as follows:
\[\begin{array}{ll}
    \free(P(v_1...v_n))       &=\{v_1,...,v_n\}\\
    
    \vspace{1mm}
    \free(\neg\varphi)      &=\free(\varphi)\\
    
        \vspace{1mm}
    \free(\varphi\wedge\psi)&=\free(\varphi)\cup \free(\psi)\\
    
        \vspace{1mm}
    \free(D_Vu) = \free(\mathbb{E}_V\psi)            &=V
\end{array}\]
\end{definition}

  \vspace{-1ex}

It might seem odd that $free(\mathbb{E}_x\varphi)=\{x\}$ while $free(\exists x\varphi)=free(\varphi)\setminus\{x\}$. The difference, however, is a matter of perspective, since the two notions of free variable coincide on their interpretation as ``a variable whose value can affect the truth of the formula''. With the preceding definitions, it is easy to show that LFD, just like FO, satisfies a property called \emph{variable locality}: any two assignments $s, t$ that agree on some set of variables $X$ agree on the truth values of all formulas whose free variables are contained in $X$ (plus all variables that locally depend on $X$ at $s$, $t$). 

The difference with FO free variables arises because the LFD modality is dual to that of FO, and displays variables whose values are kept fixed, rather than those whose values are allowed to vary. More generally, LFD shares a feature with CRS that makes it distinct from FO in its standard semantics. While in first-order formulas quantified variables are arbitrary placeholders, in LFD formulas variables have an individual character given their possibly different roles in the assignment space.\footnote{Accordingly, renaming bound variables in one model is not generally possible in LFD, though one can shift to new variables when  allowing changing the current model~\cite{LFD}.}

Standard first-order models induce special `full' dependence models  $(M,M^V)$ where \textit{all} assignments are available. On such models there are no non-trivial dependencies (atoms $D_Xy$ are always false  unless $y \in X$) and the dependence quantifier $\mathbb{E}$ collapses to the standard existential quantifier $\exists$, recapturing FO. Conversely, dependence models can also be encoded as standard models,  supporting a faithful translation into an expanded version of FO. 

Concretely, there is a translation $tr(\cdot)$ from LFD formulas to FO formulas (over an expanded signature), and a contravariant function $T$ that  bijectively maps models of $tr(\varphi)$
to dependence models satisfying $\varphi$. In a picture:

\vspace{1mm}

\begin{center}
\begin{tabular}{lll}
LFD-formula $\varphi$ & $\xrightarrow{~~~~tr~~~~}$ & first-order formula $tr(\varphi)$ \\[.5em]
dependence model $T(M)$ & $\xleftarrow[(\text{bijective})]{~~~T~~~}$ & first-order model $M$
\end{tabular}
\end{center}

\vspace{1mm}

Since we will see variants of this diagram of contravariant maps  in later sections of this paper, 
we will spell things out in detail, as it will serve as an example for our later translation schemes.

\begin{definition}{(First-Order Translation)}\\
Let $V_{LFD}=\{v_1, \ldots, v_n\}$. The translation function $tr(\cdot)$ maps LFD formulas over $V_{LFD}$ to FO formulas using
the same variables $v_1, \ldots, v_n$ and auxiliary variables $v'_1, \ldots, v'_n$, and using an additional $n$-ary relation symbol $A$.
\begin{align*}
    tr(P(u_1,\ldots,u_n)):=&\;P(u_1,\ldots,u_n) \vspace{1mm}\\
    tr(\varphi\wedge\psi):=&\;tr(\varphi)\wedge tr(\psi)\vspace{1mm}\\
    tr(\neg\varphi):=&\;\neg tr(\varphi)\vspace{1mm}\\  
    tr(D_V u):=&\;\forall \bar{v'}(A(\bar{v'})\land\bigwedge_{v\in V}(v=v'\to u=u')) \\
    tr(\mathbb{E}_V\psi):=&\;\exists\bar{z}(A(\bar{v})\wedge tr(\psi))
\end{align*}
where $\bar{v}=v_1, \ldots, v_n$, $\bar{v'}=v'_1, \ldots, v'_n$, and $\bar{z}$ enumerates $V_{LFD}\setminus V$.
\end{definition}


In this translation $A$ is a new  predicate encoding the admissible assignments on the relevant variables. 

\begin{definition}{(Model Transformation)}\\
Let $V_{LFD}=\{v_1, \ldots, v_n\}$.
For a standard model $M$ over the extended signature with the $A$ relation, $T(M)$ is the 
dependence model obtained by extracting a team $A:=\{s:V\to M\;|\;(s(v_1), \ldots, s(v_n))\in A^M\}$ on $M$ and removing the relation $A^M$ from the model. 
\end{definition}

Note that the map $T$ is bijective. Its inverse $T^{-1}$ maps a dependence model $\mathbb{M}=(M,A)$ to the standard model over the extended signature obtained by interpreting  $A$ as $\{(s(v_1), \ldots, s(v_n)) \mid s\in A\}$.

\begin{theorem} \label{thm:fo-translation-correctness}
For every LFD formula $\psi$, standard model $M$,  assignment $s:\free(\psi)\to M$
such that  $M,s\models A(v_1, \ldots, v_n)$,
we have the equivalence

\vspace{-2ex}

$$T(M), s\models\psi\quad\;\textrm{iff}\quad M,s\models tr(\psi) \land A(v_1, \ldots, v_n)\footnote{The  added conjunct $A(v_1, \ldots, v_n)$ makes sure that $s$ is an admissible assignment in $M$. We could also work this information into the compositional clauses of the translation.}$$

\end{theorem}

Since $T(\cdot)$ is bijective, Theorem~\ref{thm:fo-translation-correctness} implies, in particular, that $tr(\varphi)\land A(v_1, \ldots, v_n)$ and $\varphi$ are equi-satisfiable. 
As consequences of this translation, one can also derive an effective reduction of semantic consequence, Compactness, L\"owenheim-Skolem, recursive enumerability, all of them inherited by LFD  from FO. Other properties require additional work in reproving classical results, much as happens with the standard FO translation for modal logic \cite{BRV}.

\begin{rem}[LFD and CRS] Dependence models are the generalised  semantic structures first introduced for the algebraic version CRS of FO logic, \cite{Nemeti86}, with  $M, s \models \exists x\, \varphi$  if some available  assignment $t = ^{x} s$  in $M$  satisfies $\varphi$. Here the relation $=^{x}$ is dual to the $=_x$ used for LFD, requiring that $s, t$ agree on the values for all variables \emph{except} x. Thus CRS quantifiers, like the first-order ones, display the variables that are allowed to vary, rather than those kept fixed. A polyadic version $\exists \bar{x}\, \varphi$ is defined analogously, which, unlike in standard FO semantics, does not reduce to iterated single existential quantifiers. When the total set $V$ of variables is finite, the existential quantifiers of CRS and the existential modality of LFD are  intertranslatable, e.g., $\exists x\, \varphi$ = $\mathbb{E}_{V\setminus\{x\}} \, \varphi$. Thus, the fragment of LFD without dependence atoms amounts to a version of CRS with `dual quantifiers', more detailed discussion can be found in \cite{LFD}. Thus, some results in what follows also specialize to new connections between GF and CRS. \end{rem}


\subsection{Dependence Bisimulations, and Type Models}
\label{sec:type-models}

We need to introduce two more technical notions related to LFD that will play an important role in our proofs in subsequent sections. 

Given its  similarities with modal logic,  LFD comes with a notion of \textit{dependence-bisimulations}, which leaves the truth values of formulas invariant \cite{localdeps,Raoul}. For each set of variables $V$, dependence model $\mathbb{M}$, and assignment $s$, the \textit{dependence-closure} $D_V^s=\{u\in V_{LFD} \mid \mathbb{M},s\models D_V u\}$ is the  set of variables locally determined by $V$ at $s$. Clearly, $V\subseteq D_V^s$.
We call a set $V$ \textit{dependence-closed} at  assignment $s$ if $D^s_V=V$.
Also, for any two assignments,  $V^{s,t}=\{v\in V_{LFD}\;|\;s=_vt\}$ denotes the set of variables on which $s, t$ agree.

\begin{definition}{(Dependence Bisimulations)}
Let $(M,A),(M',A')$ be dependence models. A relation $Z\subseteq A\times A'$ is a \textit{dependence bisimulation} if for every pair $(s,s')\in Z$, the following conditions are satisfied:
\footnote{Note that dependence bisimulations $Z$ are always \textit{total}, i.e., $dom(Z)=A$ and $cod(Z)=A'.$}
\begin{description}
\item[(atom)]  $s(\bar{u})\in P^M\quad\textrm{iff}\quad s'(\bar{u})\in P^{M'}$ for all  tuples $\bar{u}$ of LFD variables
 
\item[(forth)]  for every $t\in A$, there is a $t'\in A'$ such that (i) $V^{s,t}\subseteq V^{s',t'}$,\\
  (ii) $(t,t')\in Z$ and (iii) the set $V^{s,t}$ is dependence-closed at $s'$
\item[(back)]  for every $t'\in A'$, there is a $t\in A$ such that (i) $V^{s',t'}\subseteq V^{s,t}$, \\
(ii) $(t,t')\in Z$ and (iii) the set $V^{s',t'}$ is dependence-closed at $s$
\end{description}
\end{definition}

LFD formulas do not distinguish between dependence-bisimilar models. Indeed, it was shown in~\cite{localdeps} that a first-order formula is invariant for dependence bisimulations iff it is equivalent to the $tr(\cdot)$ translation of an LFD formula.\footnote{The definition of LFD-bisimulations in \cite{localdeps} is stated differently: it does not have the dependence-closed condition (iii) in the forth and back-clauses but instead includes the requirement that $s\models D_Xy$ iff $s'\models D_Xy$ into the atomic condition. As it turns out, these definitions are equivalent, but the one in \cite{localdeps} also works for type models.}

Dependence bisimulations allow us to show that LFD is `blind' for precisely how assignments assign values to different variables. 
This is captured formally by the  \textit{distinguished model property of LFD}. 
A dependence model is \emph{distinguished} if each variable only takes values in its own range different from the ranges of all other variables. Thus for every object in the domain there is a unique variable to which it can be assigned by an admissible assignment.
\begin{proposition}\label{prop:distinguished}
There is a transformation $(\cdot)^d$ taking any dependence model $M$ into a dependence-bisimilar distinguished  dependence model $M^d$.
\end{proposition}

\begin{proof}
For each admissible assignment $s\in A$, let $s^d$ be the assignment defined by $s^d(v):=(v,s(v))$. Now take the team $A_d:=\{s^d\;|\;s\in A\}$ on the domain $\bigcup_{s\in A}s^d[V_{LFD}]$. We interpret relations on the new value domain by setting $$R^{M_d}:=\{(u_1,s(u_1)),...,(u_m,s(u_m))\;|\;(s(u_1),...,s(u_m))\in R^M\}~,$$ making it into a standard model $M_d$. Note that each $s^d:V_{LFD}\to M_d$ and hence $(M_d,A_d)$ is a dependence model which is distinguished by construction. It is now straightforward to check that the relation $\{(s,s_d)\;|\;s\in A\}$ is a dependence bisimulation between $(M,A)$ and $(M_d,A_d)$.
\end{proof}

The distinguished model property  has  to do with the lack of an equality in the language. LFD extended with equality atoms $v_i=v_j$ is  undecidable \cite{localdeps}.

Finally, we  review an abstract type model semantics, also known as `quasi-models' \cite{CRSGuardedLogics,Guardsbounds},  originally used to prove the decidability of LFD in analogy with type models for GF. One can think of type models as finite certificates for satisfiability, or as information reductions of `real' models with objects.

\begin{definition}{(Types)}\\
The \emph{closure} of an LFD formula $\psi$, denoted $Cl(\psi)$, is the smallest set
containing $\psi$ that is closed under subformulas and single negations\footnote{The single negation  $\sim(\psi)=\psi'$ if $\psi=\neg\psi'$ for some $\psi'$, and $\sim(\psi)=\neg\psi$ otherwise.} and that contains all atomic formulas of the form $D_Vu$. A subset $\Delta\subseteq Cl(\psi)$ is a \textit{type} (for $\psi$) if it satisfies the following three properties for the logical operations plus two further constraints on dependence atoms:
\begin{description}
\item[~~($\neg$-Consistency)] $\neg\chi\in\Delta$ iff not $\chi\in\Delta$
\item[~~($\land$-Consistency)] $\chi\wedge\xi\in\Delta$ iff $\chi\in\Delta$ and $\xi\in\Delta$
\item[~~($\mathbb{E}$-Consistency)] $\chi\in\Delta$ implies $\mathbb{E}_V\chi\in\Delta$ whenever $\mathbb{E}_V\chi\in Cl(\psi)$
\item[~~(Projection)] $D_Vu\in\Delta$ for all $V\subseteq V_{LFD}$ and $u\in V$
\item[~~(Transitivity)] $D_V U,D_U W\in\Delta$ implies $D_VW\in\Delta$ for all $V,U,W\subseteq V_{LFD}$  
\end{description}
\end{definition}

For  types $\Delta, \Delta'$ and sets  $V\subseteq V_{LFD}$, we will henceforth write 
$\Delta\sim_V\Delta'$ if $\{\psi\in\Delta\;|\;\free(\psi)\subseteq V\}=\{\psi\in\Delta'\;|\;\free(\psi)\subseteq V\}$.

\vspace{1ex}

\begin{definition}{(LFD Type Models)}\\
Fix an LFD formula $\psi$.
An LFD type model $\mathfrak{M}$  is a set of types (for $\psi$)  satisfying the following two conditions:

\begin{description}
\item[(witness)] $\textrm{if}\;\mathbb{E}_V\varphi\in\Delta\;\textrm{then}\;\exists\Delta'\in\mathfrak{M}\;\textrm{with}\;\Delta\sim_{D^{\Delta}_V}\Delta'\;\textrm{and}\;\varphi\in\Delta'$
\item[(universal)] $\sim_{\emptyset}\;\textrm{is the universal relation on}\;\mathfrak{M}$
\end{description}

\noindent with $D^{\Delta}_V=\{u\in V_{LFD}\;|\;D_Vu\in\Delta\}$  the `dependence-closure' of $V$ w.r.t $\Delta$.\footnote{Note that if $\Delta\sim_V\Delta'$ then $D^{\Delta}_V=D^{\Delta'}_V$ because $\free(D_Vu)=V$.}
\end{definition}
We say $\mathfrak{M}$ is a type model \emph{for} $\psi$ if there is a type $\Delta\in\mathfrak{M}$ with $\psi\in\Delta$. As it turns out, LFD type models always encode the existence of a real dependence model, so they may serve as certificates for the satisfiability of an LFD formula.

\begin{theorem}\label{thm:unraveling}
There is an `unravelling' operation $R$ such that, for each LFD type model $\mathfrak{M}$, $R(\mathfrak{M})$ is a dependence model realizing precisely the types in $\mathfrak{M}$.
\end{theorem}

In fact, there is a dependence bisimulation (in the sense of \cite{localdeps}) from $R(\mathfrak{M})$  onto $\mathfrak{M}$.
Theorem~\ref{thm:unraveling} implies that LFD is also complete w.r.t.~type models.\footnote{This analysis can  be rephrased  modal-style in terms of  `general relational models' \cite{LFD}.}

\subsection{The Guarded Fragment of First-Order Logic}
By a ``standard model'' we mean an ordinary first-order structure, and for FO formulas $\varphi$,  $\varphi(\bar{x})$ indicates that $\free(\varphi)\subseteq \{\bar{x}\}$. We will be concerned with the Guarded Fragment GF~\cite{ModLangBoundedFrag}. This is the fragment of FO consisting of formulas in which all quantification is \textit{guarded}, i.e., of the form $\exists \bar{y}(G(\bar{x},\bar{y})\wedge\psi(\bar{x},\bar{y}))$,
where $G(\bar{x},\bar{y})$ is an atomic formula and $\free(\psi)\subseteq \free(G(\bar{x},\bar{y}))=\{\bar{x},\bar{y}\}$. We will write guarded quantifiers $\exists \bar{y}(G(\bar{x},\bar{y})\wedge\varphi)$ when it is clear from context that $\varphi$ is guarded. We also use $\forall \bar{y}(G(\bar{x},\bar{y})\to\psi(\bar{x},\bar{y}))$ as 
syntactic sugar for $\neg\exists \bar{y}(G(\bar{x},\bar{y})\wedge\neg\psi(\bar{x},\bar{y}))$.
For simplicity, we will be working with the \emph{equality-free} guarded fragment. 
Our results extend to full GF with equality, cf.~Section~\ref{sec:discussion}.

For later reference, observe that the translation $tr(\cdot)$ from LFD to FO that we discussed earlier lands in GF except for the dependence atoms:  the formula $tr(D_Xy)$ uses unguarded quantification.

For the properties of the Guarded Fragment used in this paper, we refer to the earlier-mentioned sources. In particular, the satisfiability problem for GF is decidable, while the language  has a characterization as a fragment of FO in terms of invariance for \emph{guarded bisimulations}, ~\cite{ModLangBoundedFrag}.

\section{A compositional translation from GF into LFD}\label{sec:tau}



It was shown in \cite{Guardsbounds} that satisfiability in the Guarded Fragment reduces effectively to satisfiability in the logic CRS, The latter, as we mentioned before, corresponds to the fragment of LFD without dependence atoms. 
The translation achieving this used conjunctions of (a) a compositional translation for GF formulas $\varphi$, (b)  a `set-up component'  that records the requirements for a GF type model for $\varphi$, both stated as CRS formulas. 
In contrast, here, we will present a compositional translation from GF to LFD (in fact, to the fragment of the logic LFD without dependence atoms), 
allowing us to transfer properties such as Craig interpolation. Our surprisingly simple translation also neatly exhibits a conceptual difference between first-order variables and LFD variables.

\begin{rem}[Two views of variables]\label{rem:2viewsofvar}
Given the duality between the LFD modality $\mathbb{E}_X$ and the first-order quantifier $\exists X$, it is tempting to translate FO formulas to LFD by  replacing $\exists x \varphi$ by $\mathbb{E}_Y \varphi$ with $Y$  the set of all free variables of $\varphi$ except $x$. But this does not work.  Consider the unsatisfiable first-order sentence $\exists x P(x)\land\neg\exists y P(y)$. Its translation, following the recipe just described, is $\mathbb{E}_\emptyset P(x)\land \neg\mathbb{E}_\emptyset P(y)$. However, 
this LFD formula is satisfied in a dependence model $\mathbb{M}=(M,A)$ where $M$ has domain $\{a,b\}$ with just one atomic fact $P(a)$, and  every assignment in $A$ maps $x$ to $a$ and $y$ to $b$. 
The crux lies  in the fact that, in first-order logic, variables are interchangeable placeholders, whereas in LFD, 
variables have individual behavior,  at least  within a  dependence model.\end{rem}

To deal with this discrepancy, the following translation $\tau$ separates the sets of relevant variables $V_{LFD}$  used for LFD and $V_{GF}$ for GF, and then connects these explicitly via a mapping that can be seen as an abstract substitution.

\begin{definition}
Let $\varphi$ be a FO formula and $\rho:\free(\varphi)\to V_{LFD}$. The translation $\tau_{\rho}(\varphi)$ is defined  as follows, with $\rho=_X\rho'$ for $\rho(x)=\rho'(x)$ for all $x\in X$:
\[\begin{array}{lll} \displaystyle

    \vspace{1mm}
    \tau_{\rho}(P(x_1,\ldots,x_n)) &:=& P(\rho(x_1),\ldots,\rho(x_n))\\
    
    \vspace{1mm}
    \tau_{\rho}(\varphi\land\psi) &:=& \tau_{\rho}(\varphi)\land\tau_{\rho}(\psi) \\
    
        \vspace{1mm}
    \tau_{\rho}(\neg\varphi) &:=& \neg\tau_{\rho}(\varphi) \\
    
        \vspace{1mm}
    \tau_{\rho}(\exists\bar{y} ~ \varphi(\bar{x},\bar{y})) &:=& 
      \displaystyle\bigvee_{\overset{\rho':\{\bar{x},\bar{y}\}\to V_{LFD}}{\rho\;=_{\bar{x}}\;\rho'}}\mathbb{E}_{\rho(\bar{x})} ~ \tau_{\rho'}(\varphi)
\end{array}\]

\vspace{-1ex}
\end{definition}

With this new compositional translation, and assuming that $V_{LFD} = \{v_1, v_2\}$,
the above FO formula
$\exists x P(x)\land\neg\exists y P(y)$  translates to the LFD-formula
$(\mathbb{E}_\emptyset P(v_1) \lor \mathbb{E}_\emptyset P(v_2)) \land \neg(\mathbb{E}_\emptyset P(v_1) \lor \mathbb{E}_\emptyset P(v_2))$,  which, indeed, is unsatisfiable. 
Even so,  $\tau$ is  not an effective satisfiability reduction 
 from FO  to LFD, since LFD is decidable. But as we will show, $\tau$ preserves satisfiability for formulas in the Guarded Fragment. 
 
Our next step is to define a corresponding operation on models. The overall scheme will look like this:

\vspace{1ex}

\begin{center}
\begin{tabular}{lll}
GF-formula $\varphi$ & $\xrightarrow{~~~~~\tau~~~~~}$ & LFD formula $\tau_\rho(\varphi)$ \\[.5em]
standard model $G(\mathbb{M})$ & $\xleftarrow[\text{(surjective)}]{~~~G~~~}$ & dependence model $\mathbb{M}$
\end{tabular}
\end{center}

\vspace{1ex}
%
 %
The model transformation $G$, consists of throwing away 'unnamed' facts.

\begin{definition}
Given any dependence model $\mathbb{M}=(M,A)$, the standard model $G(\mathbb{M})$ has for its domain $dom(M)$ and for its predicate interpretations $R^{G(\mathbb{M})}:=\{\bar{m}\in R^{M} \mid \{\bar{m}\}\subseteq s[V_{LFD}]$ for some $s\in A\}$, where $s[V_{LFD}]$ is the image of $s:V_{LFD}\to M$, i.e. $s[V_{LFD}]=\{s(v)\;|\;v\in V_{LFD}\}$.
\end{definition}

The transformation $G$ is surjective because each standard model $M$ is the $G$ image of its matching `full' dependence model $F(M)=(M,M^{V_{LFD}})$. Note that, while 
$G(F(M))=M$, in general $F(G(\mathbb{M}))$ can be very different from $\mathbb{M}$, since $F(G(\mathbb{M}))$ is always full. A further complication is that, while the operation $G$ is defined on all dependence models, the crucial equivalence theorem that connects $\tau$ with $G$, 
which we will present next, applies only for \emph{distinguished} dependence models. 
We will address this issue separately.

\begin{theorem}\label{thm:distinguished-correctness}
For every distinguished dependence model $\mathbb{M}=(M,A)$, assignment $s\in A$ and GF formula $\varphi$ with map $\rho:\free(\varphi)\to V_{LFD}$ we have

\vspace{-1ex}

$$\mathbb{M},s\models\tau_{\rho}(\varphi)\qquad\textrm{iff}\qquad G(\mathbb{M}),s\circ \rho\models\varphi$$
\end{theorem}

\begin{proof}
By induction on the complexity of the formula $\varphi$. The atomic case says that $(M,A),s\models P\rho(x_1)...\rho(x_n)$ holds iff $G(M,A),s\circ\rho\models Px_1...x_n$. This holds because, by our definition, the fact $P(s(\rho(x_1)), ..., s(\rho(x_n)))$ does not get removed from $(M, A)$ by the transformation $G$. The Boolean cases are routine, so the only important case to analyze is that of guarded quantification.

\vspace{1ex}

From left to right, suppose  $(M,A),s\models\tau_{\rho}(\exists \bar{y}(G(\bar{x},\bar{y})\wedge\varphi)$ (where $\rho:\{\bar{x}\}\to V_{LFD}$). This means that there is some $\rho':\{\bar{x}, \bar{y}\}\to V_{LFD}$ extending $\rho$ such that $(M,A),s\models\mathbb{E}_{\rho[\bar{x}]}(\tau_{\rho'}(G(\bar{x},\bar{y})\wedge\varphi))$, since one of the disjuncts in our translation clause is true. By the LFD semantics, there is then some witnessing assignment $t\in A$ with $s=_{\rho[\bar{x}]}t$ and $t\models \tau_{\rho'}(G(\bar{x},\bar{y})\wedge\varphi)$. By the inductive hypothesis, it now follows that $G(M,A),\;t\circ\rho'\models G(\bar{x},\bar{y})\wedge\varphi$. But since $\rho[\bar{x}]=\rho'[\bar{x}]$ and  $s=_{\rho[\bar{x}]}t$ also $s\circ\rho=_{\bar{x}}t\circ\rho'$. The GF semantics yields that $G(M,A),\;s\circ\rho\models\exists \bar{y}(G(\bar{x},\bar{y})\wedge\varphi)$ where the objects $(t\circ\rho')[\bar{y}]$ witness the bound variables $\bar{y}$.

\vspace{1mm}

From right to left, we need to appeal to the guard in the quantification.  Suppose that $G(M,A),s\circ\rho\models\exists \bar{y}(G(\bar{x},\bar{y})\wedge\varphi)$. By the standard FO semantics, there are objects $\bar{m}=(m_1,...,m_n)$ s.t. $G(M,A),(s\circ\rho) \,[\bar{m}/\bar{y}] \models G(\bar{x},\bar{y})\wedge\varphi$. By distinguishedness, for each $m\in\bar{m}$ there is a \textit{unique} LFD-variable $v_m$ s.t. there is some admissible assignment $s\in A$ with $s(v_{m})=m$. Now consider the true guard fact $G((s \circ \rho)(\bar{x}), \bar{m})$ (which is a fact of $G(M,A)$). Given the uniqueness of variable names in $M$, this fact must have been witnessed by some available assignment $t\in A$ using some atomic formula $G(\rho(\bar{x}), \bar{v_m})$, where $t \circ \rho(x) = s \circ \rho(x)$ for all $x \in \bar{x}$. It follows that $s =_{\rho(\bar{x})} t$. Now let $\rho':\{\bar{x},\bar{y}\}\to V_{LFD}$ extend $\rho$ by mapping $y_i\mapsto v_{m_i}$ for each $1\leq i\leq n$. It follows that (i) $s =_{\rho(\bar{x})} t$  in the model $(M, A)$, and (ii) $t \circ \rho' =_{\{\bar{x},\bar{y}\}} s\circ\rho \,[\bar{m}/\bar{y}]$ as assignments in the model $G(M, A)$. In particular then, since the truth  of a FO-formula under an assignment depends only on what the assignment maps its free variables to,  $G(M,A), t\circ\rho'  \models G(\bar{x},\bar{y})\wedge\varphi$. Applying the inductive hypothesis, we see that $M, t \models  \tau_{\rho'}(G(\bar{x},\bar{y})\wedge\varphi)$ and hence  $M, s \models  \mathbb{E}_{\rho(\bar{x})}(\tau_{\rho'}(G(\bar{x},\bar{y})\wedge\varphi)$. This is one of the disjuncts in the translation  $\tau_{\rho}(\exists \bar{y}(G(\bar{x},\bar{y})\wedge\varphi))$, and we are done. \end{proof}

As noted before, the map $G$ from dependence models to standard models is 
surjective. However, not every standard model is the $G$-image
of a \emph{distinguished} dependence model. This is an issue, because Theorem~\ref{thm:distinguished-correctness} applies only to distinguished dependence models.
Fortunately, we have the following:

\begin{theorem}\label{thm:gf-bisimilar-distinguished}
Let the number of variables in $V_{LFD}$ be at least as great as the maximum arity of
relations in the signature.
Then, for every standard model $M$ and guarded assignment $s$, 
there is a distinguished dependence model $\mathbb{M}$  and an 
admissible LFD assignment $t$ in $\mathbb{M}$ such
that $(G(\mathbb{M}), t\circ\rho)$ is GF-bisimilar to $(M,s)$,  for some map $\rho$.
\end{theorem}

\begin{proof}
Let $\mathbb{M}=(M',A)$ where $M'$, $A$ are as follows:
\begin{itemize}
\item $Dom(M') = Dom(M) \times V_{LFD}$, and
\item $((m_1, u_1) \ldots, (m_k, u_k))\in R^{M'}$ iff 
    $(m_1, \ldots, m_k)\in R^M$ and for all $i,j\leq k$, if $m_i\neq m_j$ then $u_i\neq u_j$.
\item $A\subseteq Dom(M')^{V_{LFD}}$ consist of all assignments that
map each LFD-variable $v$ to a pair of the form $(m, v)$ for $m\in dom(M)$ some object in the old domain.
\end{itemize}
It is clear from the construction that $\mathbb{M}$ is distinguished. 

\medskip
\medskip

Since $s$ is a guarded assignment, there exists a tuple $(m_1, \ldots, m_k)\in R^{M}$ for 
some relation $R$, such that $s$ maps every variable to an element of this tuple. 
We know that $|V_{LFD}|\geq k$. Let $\rho$ be an arbitrary mapping from FO variables
to LFD variables, such that $s(x)\neq s(y)$ iff $\rho(x)\neq \rho(y)$. Finally, let $t$
be an LFD assignment that maps each $\rho(x)$ to the pair $(s(x),\rho(x))$. In this
way, we have that, for each FO variable $x$, $(t\circ \rho)(x) = (s(x),\rho(x))$. Furthermore,
by construction, $t$ is an admissible assignment for $\mathbb{M}$.

It remains only to show that $(G(\mathbb{M}, t\circ \rho)$ is GF-bisimilar to $(M,s)$.
The GF bisimulation in question consists of all partial functions $f:Dom(M')\to Dom(M)$ 
such that (i) $dom(f)$ is a guarded subset of $dom(M')$, and (ii) $f$ is the natural 
projection on its domain, i.e., $f((m,v)=m$. It can be easily verified that this 
satisfies all requirements of a GF-bisimulation. 
\end{proof}


We state one immediate consequence, further implications of our translations will be discussed in Section 5 below.

\begin{corollary}\label{cor:tau}
Let the number of variables in $V_{LFD}$ be at least as great as the maximum arity of
relations in the signature.
Then, a GF-formula $\varphi$ is satisfiable on standard models iff, for some function $\rho:\free(\varphi)\to V_{LFD}$,  $\tau_{\rho}(\varphi)$ is satisfiable on dependence models.
\end{corollary}

\begin{proof}
From right to left, if the LFD-formula $\tau_{\rho}(\varphi)$ is satisfiable, it is also satisfiable 
(by Proposition~\ref{prop:distinguished})
on a pointed distinguished model $(M,A),s$, and so, by Theorem 3.4 we also have $G(M,A),s\circ\rho\models\varphi$. 

From left to right, suppose  $M, t\models\varphi$. Since every GF-formula is a Boolean combination of self-guarded formulas, we can assume without loss of generality that $t$ is a guarded assignment. 
It then follows  from Theorem~\ref{thm:gf-bisimilar-distinguished} together with 
Theorem~\ref{thm:distinguished-correctness} and the invariance of GF-formulas under GF bisimulations, that, for some map $\rho$, $\tau_\rho(\varphi)$ is satisfied in a (distinguished) dependence model.
 \end{proof}

%
%
%

Given the compositional nature of our translations $\tau_\rho$ on the Boolean operations, this  immediately also implies a faithful reduction of semantic consequence in GF to semantic consequence in LFD.


\section{A translation from LFD to GF}

We now present a satisfiability-preserving translation from LFD to GF.
Note that the translation $tr(\cdot)$ from LFD to FO given in Section~\ref{sec:prel}
yields FO formulas that are, in general, not guarded. For example, 
if $V_{LFD}=\{x,y,z\}$, then
$tr(D_{x}y)$ is the first-order formula $\varphi(x)= \forall y'z' (A(x,y',z')\to y=y')$, which uses unguarded quantification. Note, though, that
the $tr$-translation of an LFD formula without dependence atoms is guaranteed to be
a GF formula.
In this section, we present a satisfiability-preserving translation from full LFD, including the dependence atoms, to GF.
This translation works by replacing dependence atoms by relational atoms using fresh relations, and adding  guarded conditions that force these relations to behave like dependence atoms.\footnote{This same trick with additional predicates is used in proving the finite dependence model property for LFD by an appeal to Herwig's theorem~\cite{Raoul}.}

\begin{definition} ($\widehat{\mathbf{S}}$)
LFD was defined relative to a relational
signature $\mathbf{S}$ and a finite set of variables $V_{LFD}$. 
We now denote by $\widehat{\mathbf{S}}$ the relational signature that extends $\mathbf{S}$ with an $n$-ary relation $A$, for $n=|V_{LFD}|$, and with
a fresh $k$-ary relation symbol $R^{X,Y}$ for each pair of subsets $X,Y\subseteq V_{LFD}$, 
where $k=|X|$.
\end{definition}


The overall scheme of our translation will look like this:

\medskip

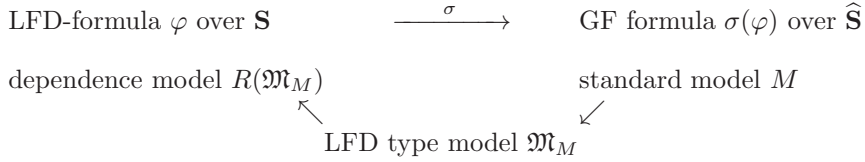
\begin{center}
\begin{tabular}{l@{}c@{}l}
LFD-formula $\varphi$ over $\mathbf{S}$ & $\xrightarrow{~~~~~\sigma~~~~~}$ & GF formula $\sigma(\varphi)$ over $\widehat{\mathbf{S}}$ \\[1em]
dependence model $R(\mathfrak{M}_M)$ & & standard model $M$ \\
\hfill $\nwarrow$ && $\swarrow$ \\
& LFD type model $\mathfrak{M}_M$
\end{tabular}
\end{center}

\medskip
In particular, our operation on models takes a detour through type models. In this diagram,
 $R$ is the mapping from type models to dependence models given by Theorem~\ref{thm:unraveling}.

Next we define the matching formula translation $\sigma$.

\begin{definition}\label{def:sigmatransLFDtoGF}
Let $V_{LFD}=\{v_1, \ldots, v_n\}$.
For an LFD-formula $\psi$ over a signature $\mathbf{S}$, we define $\sigma(\psi)$ as the following GF-formula over the signature
$\widehat{\mathbf{S}}$: 
$$\sigma(\psi)~~:=~~ tr^\bullet(\psi)\wedge setup(\psi) \land A(v_1, \ldots, v_n)$$

\vspace{-1ex}
\noindent where
\begin{itemize}
    \item $tr^\bullet(\psi)$ is defined inductively in the same way as $tr(\psi)$, with the only difference  that $tr^\bullet(D_V U)) = R^{V,U}(\bar{v})$. 
    (Here, for the sake of presentation, we assume a fixed, canonical ordering on variables, 
    so that it is clear which element of the tuple $\bar{v}$ corresponds to which element of the set $V$.)
    
\item $setup(\psi)$ is the conjunction of all GF-sentences of the form

\vspace{-1ex}

\end{itemize}
\begin{align*}
    &\quad\forall v_1, \ldots, v_n(A(v_1, \ldots, v_n)\to \bigwedge_{u\in V}R^{V,u}(\bar{v}))\\
    &\quad\forall v_1, \ldots, v_n(A(v_1, \ldots, v_n)\to (R^{V,U}(\bar{v})\wedge R^{U,W}(\bar{u})\to R^{V,W}(\bar{v}))\\
    &\quad\forall v_1, \ldots, v_n(A(v_1, \ldots, v_n)\to[R^{V,U}(\bar{v})\land tr^\bullet(\mathbb{E}_V\xi)\to
        tr^\bullet(\mathbb{E}_{V\cup U}\xi)])
\end{align*}
where $V,U,W\subseteq V_{LFD}$; $\bar{v}, \bar{u}, \bar{w}$ are canonical enumerations of these sets and $\xi\in Cl(\psi)$ is a formula whose main connective
is not conjunction or negation.
\end{definition}

The formulas in the ``set up'' part of the above translation can be viewed as instances of the Projection, Transitivity and Transfer axioms in the axiomatization of LFD given in~\cite{LFD}.


Next, we define our operation $H$, which takes a standard model and produces an LFD type model.

\begin{definition}\label{def:type-model-construction}
Let $V_{LFD} = \{v_1, \ldots, v_n\}$ and let $\psi$ be an LFD formula.
For a standard model $M\models setup(\psi)$, we define $\mathfrak{M}_M^\psi$ as:
\[\mathfrak{M}_M^\psi :=\{\textrm{type}^\psi(M,s) \mid \text{$s:V_{LFD}\to Dom(M)$ such that $M,s\models  A(v_1, \ldots, v_n)$} \}\]
with $\textrm{type}^\psi(M,s) = \{\varphi\in Cl(\psi)\mid M,s\models tr^\bullet(\varphi)\}$, 
where $Cl(\psi)$ is the closure of $\varphi$, as defined in Section~\ref{sec:type-models}.
\end{definition}

\begin{lemma}\label{lem:setup}
If $M\models setup(\psi)$, then $\mathfrak{M}_M^\psi$  is an LFD type model. 
\end{lemma}

\begin{proof}
First, we show that each $\Delta\in\mathfrak{M}_M^\psi$ is a type. The $\neg$-consistency, 
$\land$-consistency, and $\mathbb{E}$-consistency properties are immediate from the
construction of $\mathfrak{M}_M^\psi$ and the definition of $tr^\bullet$. Projection and transitivity 
hold by virtue of the corresponding conjuncts of $setup(\psi)$. 

The \emph{universal} property, i.e., that 
$\sim_{\emptyset}$ is the universal relation on $\mathfrak{M}$, also follows 
immediately from the construction of $\mathfrak{M}_M^\psi$ using the fact that
$tr^\bullet$-translation of an LFD-sentence is a GF-sentence.

Finally, for the \emph{witness} property, suppose that 
$\mathbb{E}_V\varphi\in \Delta$. Let $s$ be a witnessing assignment for $\Delta$.
Then, in particular, $M,s\models A(v_1, \ldots, v_n)\land tr^\bullet(\mathbb{E}_V\varphi)$.
It follows by the definition of $tr^\bullet$ that there is an assignment $t$
with $t=_{V} s$, such that $M,t\models A(v_1, \ldots, v_n)\land tr^\bullet(\varphi)$.
Let $\Delta'\in\mathfrak{M}_M^\psi$ be the type obtained from $t$ as in Definition~\ref{def:type-model-construction}. Then $\varphi\in \Delta'$. Moreover, since $t=_V s$, and the translation $tr^\bullet$ preserves the 
free variables of the formula, we have that $\Delta\sim_{V}\Delta'$. In order to satisfy the
witness property of type models, we must show something slightly stronger, namely
that $\Delta\sim_{D^\Delta_V}\Delta'$.
This can be shown in two steps: (i) it is easy to see that
$D^{\Delta'}_V = D^\Delta_V$ by virtue of the fact that $\free(tr^\bullet(D_{V}U))=V$; (ii)
the third conjunct of $setup(\psi)$ now allows us to lift the $\Delta\sim_V\Delta'$ relation to the
$\Delta\sim_{D^\Delta_V}\Delta'$. To see this, first note that, by the 
$\neg$-consistency and $\land$-consistency properties of types, it is enough
to show that $\Delta$ and $\Delta'$ agree on ``non-decomposable'' formulas
with free variables in $D^\Delta_V$, where we call a formula decomposable if
its main connective is a conjunction or a negation. 
Let $\xi\in\Delta$ with $Y=\free(\xi)\subseteq D^{\Delta}_V$,
and assume that $\xi$ is non-decomposable. 
Then  $M,s\models tr^{\bullet}(\xi)$ and $M,s\models R^{V,Y}(\bar{v})$.
By existential generalisation (cf.~the definition of $tr^\bullet$) we get
that $M,s\models tr^{\bullet}(\mathbb{E}_{V}\xi)$. 
Since $s=_Vt$, this implies $M,t\models tr^{\bullet}(\mathbb{E}_V\xi)$ and
$M,t\models R^{V,Y}(\bar{v})$ as well. Since we know that $M,t\models setup(\psi)$, we can now apply the `transfer' condition to get that $M,t\models tr^{\bullet}(\mathbb{E}_{V\cup Y}\xi)$ which implies that $M,t\models\xi$ as $\free(\xi)\subseteq V\cup Y$. Hence $\xi\in\Delta'$.
The converse direction follows from a symmetric argument.
\end{proof}

We now obtain our dependence model $H(M)$ by applying the unraveling operation $R$
from Section~\ref{sec:type-models} to $\mathfrak{M}_M^\psi$.

\begin{definition}
Fix an LFD formula $\psi$ over signature $\mathbf{S}$,  let $M$ be a
 standard model for $\widehat{\mathbf{S}}$ s.t. $M\models setup(\psi)$, and let $s:V_{LFD}\to dom(M)$ such that $M,s\models A(v_1, \ldots, v_n)$. 
 \begin{enumerate}
 \item We define $H(M)$ as $R(\mathfrak{M}_M^\psi)$.
 \item We define $H(s)$
as an (arbitrarily chosen) admissible LFD assignment on $H(M)$ that realizes
 $\textrm{type}^\psi(M,s)$ in $H(M)$.
\end{enumerate}
\end{definition}

While not reflected in our notation, the model $H(M)$ depends on $\psi$.
Also, note that  $\textrm{type}^\psi(M,s)\in \mathfrak{M}^\psi_{M}$, and, by Theorem~\ref{thm:unraveling},
 $H(M)$ realizes every type in $\mathfrak{M}^\psi_{M}$, making $H(s)$ indeed well-defined.

\begin{theorem}\label{thm:lfd-to-gf-correctness}
Fix an LFD formula $\psi$ over signature $\mathbf{S}$, and let $M$ be 
standard model over the signature $\widehat{\mathbf{S}}$ such that $M\models setup(\psi)$,
and let $s:V_{LFD}\to Dom(M)$,  such that $M,s\models A(v_1, \ldots, v_n)$. Then, for all formulas
$\varphi\in Cl(\psi)$,
\[H(M),H(s)\models\varphi\qquad\textrm{iff}\qquad M,s\models tr^\bullet(\varphi)\]
\end{theorem}

\begin{proof}
If $H(M),H(s)\models\varphi$, then by the definition of $H(s)$, we have that
$\varphi\in type^{\psi}(M,s)$, i.e., $M,s\models tr^\bullet(\varphi)$.
Conversely, if $M,s\models tr^\bullet(\varphi)$, then $\varphi\in type^{\psi}(M,s)$. 
Hence, again by the definition of $H(s)$, we have that $H(M),H(s)\models\varphi$. \end{proof}

\begin{corollary}\label{cor:entailment}
\begin{enumerate}
\item An LFD formula $\psi$ is satisfiable  iff its GF-translation $\sigma(\psi)$ is satisfiable.
\item An LFD  entailment $\varphi\models\psi$  is valid iff the GF-entailment
   $setup(\varphi)\land tr^\bullet(\varphi)\land A(v_1, \ldots, v_n)\models setup(\psi)\to tr^\bullet(\psi)$ is valid.
\end{enumerate}
\end{corollary}

\begin{proof}
From left to right, if $(M,A),s\models\psi$, then let $M'$ be the $\widehat{\mathbf{S}}$-expansion
of $M$ obtained by setting $A^{M'}$ as $\{(s(v_1), \ldots, s(v_n))\mid s\in A\}$, and 
each $(R^{V,U})^{M'}=\{s(\bar{v}) \mid (M,A),s\models D_VU(\bar{v})\}$. A straightforward
formula induction shows that, for all LFD formulas $\varphi$, $(M,A),s\models\varphi$ iff
$M',s\models A(v_1, \ldots, v_n)\land tr^{\bullet}(\varphi)$.
Furthermore, $M'\models setup(\psi)$ since all auxiliary predicates have their intended interpretation.
Therefore, we have that $M',s\models\psi$.
The converse direction follows immediately from Theorem~\ref{thm:lfd-to-gf-correctness}.

(ii) follows from (i), because (as careful inspection shows) 
$setup(\varphi\land\neg\psi)$ is logically equivalent to $setup(\varphi)\land setup(\psi)$
(to see this, note that, if $\varphi\land\neg\psi$ has a subformula that is 
non-decomposible, then the latter must be a subformula of $\varphi$ or of $\psi$)
and the fact that  $tr^\bullet$  commutes with the Boolean connectives.
\end{proof}

\section{Transfer results from our translations}

In addition to shedding light on the relationship between LFD and GF, 
our translations can also be used to
transfer results from GF to LFD and vice versa. 
We  give  several examples connecting  known results, but also a new result, viz. a complete complexity analysis of LFD.

As a warming up, we reprove a known result, namely the finite model property for GF~\cite{Graedelguards},
Recall that a logic has the finite model property (FMP) if every
satisfiable formula has a finite model.
It has been long known that GF has the FMP~\cite{Graedelguards},
and FMP was only recently shown to hold for LFD~\cite{Raoul}. Therefore, 
the following transfer fact does not
give us a new result, but it does provide an illustration of the value of our translations.

\begin{theorem}[FMP for GF]
GF has the finite model property.
\end{theorem}

\begin{proof}
Let $\varphi$ be a satisfiable guarded formula. By corollary~\ref{cor:tau}, since $\varphi$ is satisfiable on standard models, it follows that $\tau_{\rho}(\varphi)$ is satisfiable on dependence models, for some $\rho:\free(\varphi)\to V_{LFD}$ (for  $V_{LFD}$ sufficiently large).
By the finite model property of LFD \cite{Raoul}, there is a 
finite dependence model $\mathbb{M}$ satisfying $\tau_{\rho}(\varphi)$ at some $s\in A$. We can take $\mathbb{M}$ to be distinguished by Proposition~\ref{prop:distinguished} (which preserves finiteness). In fact, inspection of the proof in \cite{Raoul} shows that its construction via Herwig's theorem already ensures distinguishedness. By Theorem~\ref{thm:distinguished-correctness}, then, $G(\mathbb{M})$ is a finite model of $\varphi$.
\end{proof}

\begin{theorem}[Complexity of LFD satisfiability]\label{thm:sat-transfer}
For a finite set of LFD-variables $V_{LFD}$, 
the satisfiability problem for LFD-formulas in $V_{LFD}$ is ExpTime-complete.
The same problem is 2ExpTime-complete if $V_{LFD}$ is considered as part of the input.
\end{theorem}

\begin{proof}
Upper bound: by reduction to GF.
The satisfiability problem for GF is 2ExpTime-complete~\cite{Graedelguards}. More precisely,
the satisfiability of a GF formula with $k$ FO-variables can be checked in time $2^{O(|\varphi|\cdot k^k)}$ (cf.~\cite{MarxVenema}).
A careful analysis of our translation from LFD to GF shows 
it can be performed in time $O(|\varphi|)\cdot 2^{O(k)}$, where $k=|V_{LFD}|$. Indeed,  
computing the modified first-order translation $tr^{\bullet}(\varphi)$ can be done in linear time, so only the set-up part brings all the complexity. The formula $setup(\varphi)$ consist of $2^k$ many conjuncts encoding a couple of projection axiom instances on the proxy-dependence atoms $R^{X,y}$ as well as $2^{3k}$ many conjuncts encoding a transitivity axiom instance and  $O(|\varphi|)\cdot 2^{2k}$ many conjuncts encoding a transfer axiom instance for some `indecomposable' formulas in the closure $Cl(\varphi)$. Putting this together, we obtain that the satisfiability problem for LFD is in ExpTime
if $V_{LFD}$ is treated as fixed in the complexity analysis, and that it is in 
2ExpTime if $V_{LFD}$ is part of the input (due to the fact that, even though the translation from  LFD to GF is exponential, it does not increase the number of variables).

Lower bound: by reduction from GF. 
Observe that our translation from GF to LFD is exponential due to the disjunction in the translation clause for the guarded quantifiers. Exponential in the nesting depth of (polyadic) guarded-quantifications, to be precise. 
For this reason, it is not immediately clear that  complexity upper bounds for GF transfer to LFD as is. However, we can make use of a
result from~\cite{Graedelguards}, according to which there is a 
satisfiability-preserving translation from GF-formulas to a ``Scott-normal form'' (which can be performed in polynomial time and without increasing the number of variables), where the normal form uses only two levels of polyadic guarded quantification. Our translation is 
polynomial when applied to GF-formulas in this normal form.
As we mentioned earlier, the satisfiability problem for GF-formulas
with a bounded number of variables is 
ExpTime-complete, and with an unbounded number of variables it is 
2ExpTime-complete. Since our translation from GF to LFD does not
require more LFD variables than the number of variables in the input
formula, this establishes our lower bounds.
\end{proof}

In fact, we can improve this result a little.
By the \emph{monadic fragment} of LFD we mean LFD-formulas that only use
unary relation symbols.

\begin{theorem}[Complexity of the monadic fragment of LFD]
Fix a finite set $V_{LFD}$ of LFD-variables with $|V_{LFD}|\geq 2$
The satisfiability problem for the monadic fragment of LFD with $V_{LFD}$ is ExpTime-complete.
\end{theorem}

\begin{proof}[sketch]
The upper bound follows from Theorem~\ref{thm:sat-transfer}.
For the lower bound, we reduce from the satisfiability problem
of the basic modal logic $\textbf{K}$ extended with the global modality
(which we will denote as $\textbf{K}+U$).
This logic is known to be ExpTime-complete (cf.~\cite{BRV}).
By a simple encoding trick due to Halpern and Vardi~\cite{HalpernVardi1989}, 
the satisfiability problem for $\textbf{K}+U$ reduces to the satisfiability 
problem for the multi-modal logic $\textbf{S5}_2 + U$ 
(where $\textbf{S5}_2$ is the bi-modal ``fusion'' logic
that has two S5 modality without interaction axioms). The coding trick
in question consists of replacing every occurrence of $\Diamond$ by 
$\Diamond_1\Diamond_2$ where $\Diamond_1$ and $\Diamond_2$ are the two $\textbf{S5}$-modalities (cf.~\cite{HalpernVardi1989} for the proof
that this preserves satisfiability). 
The satisfiability problem for $\textbf{S5}_2+U$, 
finally, embeds straightforwardly into the monadic fragment of LFD
(even without using dependence atoms): each proposition letter becomes
a unary predicate, $\Diamond_1$ becomes $\mathbb{E}_x$ and $\Diamond_2$ becomes $\mathbb{E}_y$
for $x,y\in V_{LFD}$ two distinct LFD-variables, whereas the global modality
(in its existential form) becomes $\mathbb{E}_{\emptyset}$.
\end{proof}

\begin{remark}
Our translation from GF to (the dependence-atom-free fragment of)
LFD can also be composed with the $tr$ translation from LFD to GF
to obtain a SAT-reduction from GF to its ``universal guard'' fragment
(i.e., GF-formulas with a single guard predicate that occurs only in 
guard position).%
\footnote{In fact, to its 
fragment where all quantifiers are guarded by the same atom $A(v_1,\ldots,v_n)$.}
\end{remark}

As our final example, we prove that LFD has Craig interpolation, offering a model-theoretic alternative  to the sequent calculus-based proof  in \cite{LFD}.

\begin{theorem}[Craig interpolation for LFD]
For every valid LFD implication
$\models\varphi\to\psi$, there is an LFD formula $\vartheta$ such that
\begin{enumerate}
\item $\models\varphi\to\vartheta$, 
\item $\models\vartheta\to\psi$, and 
\item all relation symbols
occurring in $\vartheta$ occur both in $\varphi$ and in $\psi$.
\end{enumerate}
\end{theorem}

\begin{proof} [sketch] Let $\varphi\to\psi$ be valid in LFD. By Corollary~\ref{cor:entailment}, the implication
$(setup(\varphi)\land tr^\bullet(\varphi)\land A(v_1, \ldots, v_n))\to ((setup(\psi)\land A(v_1, \ldots, v_n))\to tr^\bullet(\psi))$ is valid in GF.
Note that, in this GF formula, all quantification is guarded by $A$. 
Furthermore, the antecedent and the negation of the consequent are self-guarded (meaning that the 
free variables are guarded by an atomic conjunct), and the free variables
of the consequent form a subset of the free variables of the antecedent 
(viz.~$\{v_1, \ldots, v_n\}$).
Hence, by the weak Craig interpolation theorem~\cite[Thm.~4.5]{InterpolGF}, there is
a GF-formula $\chi$ such that
\begin{enumerate}
\item $\models (setup(\varphi)\land tr^\bullet(\varphi)\land A(v_1, \ldots, v_n))\to\chi$,
\item $\models\chi\to  ((setup(\psi)\land A(v_1, \ldots, v_n))\to tr^\bullet(\psi))$, and
\item All relation symbols occurring in $\chi$ except possibly for $A$, 
  occur both in $(setup(\varphi)\land tr^\bullet(\varphi))$ and in $ ((setup(\psi)\land A(v_1, \ldots, v_n))\to tr^\bullet(\psi))$.
  \item The free variables of $\chi$ belong to $V_{LFD}$.
\end{enumerate}
It suffices only to translate $\chi$ back to LFD. This is not entirely trivial, since
$\chi$ is now a GF-formula over the expanded signature $\widehat{\mathbf{S}}$.
That is, it will in general contain the  $A$ and $R^{V,U}$ relations. 
Fortunately, 
it turns out that the translation $\tau$ from GF to LFD that we gave in Section~\ref{sec:tau} 
can be extended in a straightforward way to the full signature $\widehat{\mathbf{S}}$, as 
follows:

\vspace{-3ex}
\[ \begin{array}{lll}
\tau_\rho(A(\bar{x})) &=& \text{$\top$ if $\rho(\bar{x})$ is precisely the sequence $v_1, \ldots, v_n$; $\bot$ otherwise}   \vspace{-1mm} \\[1em]
\tau_\rho(R^{U,V}(\bar{u})) &=& \text{$D_{\rho[U]}\rho[V]$}
\vspace{-2ex}
\end{array}\]
%

\vspace{1mm}

For any dependence model $\mathbb{M}=(M,A)$, let $\widehat{\mathbb{M}}$ be the 
standard model that is the $\widehat{\mathbf{S}}$-expansion of $M$ obtained in the
natural way (as we did in the proof of Corollary~\ref{cor:entailment}). Then,
it can be shown that, with the above modification to $\tau$, 

\vspace{1mm}

\begin{enumerate}
\item [(a)] For all distinguished $\mathbb{M}$ and admissible $s$, and for all
GF-formulas $\xi$ over $\widehat{\mathbf{S}}$, we have
 $\widehat{\mathbb{M}},s\circ\rho\models\xi$ iff $\mathbb{M},s\models \tau_\rho(\xi)$.  
\item[(b)] For every LFD formula $\xi$, $\tau_\rho(tr^{\bullet}(\xi))$ is equivalent to $\xi$
  over distinguished dependence models, and hence over all dependence models, if we take $\rho$ to be the identity map on $V_{LFD}$.
\end{enumerate}

\vspace{1mm}

Putting this all together, it follows that $\vartheta := \tau_\rho(\chi)$ (under the above modified translation function $\tau$, and with $\rho$ the identity map on $V_{LFD}$)
is a Craig interpolant for our original implication $\varphi\to\psi$ on distinguished dependence models, and hence, by Proposition~\ref{prop:distinguished}, on all dependence models.
\end{proof}

LFD  also has Craig interpolants that only use shared variables, \cite{LFD}. To obtain this further information, the preceding  analysis should be refined.


\section{Further Directions}\label{sec:discussion}

\paragraph{Language extensions} One natural question is if our results extend to richer languages than the ones considered here. Here is an example. 

The translation $\tau_{\rho}$ from GF into LFD presented in Section \ref{sec:tau} can be extended to deal with \emph{identity atoms} $x = y$ between GF variables $x, y$. It suffices to add  two clauses: (a) $\tau_{\rho}(x = y) = \top$ if $\rho(x) = \rho(y)$ [the mapping to LFD variables enforces identity throughout], (b) $\tau_{\rho}(x = y) = \bot$ if $\rho(x) \neq \rho(y)$ [the values will always be distinct in distinguished dependence models]. Thus GF(=), too, translates compositionally into the decidable logic LFD without identity.

Another natural extension from an LFD perspective is adding \emph{local independence atoms} $I_Xy$ saying that fixing the local values of the variables $X$ at assignment $s$ in the current dependence model puts no constraint on $y$: which can still take any value in its range on the $X$-restricted subset of assignments. Adding atoms $I_Xy$ to LFD results in an undecidable logic \cite{LFD}. Still, even from the original motivation for CRS, analyzing true FO quantifiers in terms of their independence behavior versus the dependence behavior of CRS quantifiers makes sense, suggesting a study of richer fragments of FO.\footnote{Other natural language extensions to explore would introduce \emph{fixed-point operators}, \cite{Graedelguards}, a device whose power has not yet been studied for the dependence logic LFD.}




\vspace{1ex}

\paragraph{Translation patterns} Our three translations show some general patterns. 

Translating LFD into the FO language via $tr(.)$ was compositional, while there was an inverse model transformation $T$ supporting the usual contravariant translation equivalence $M, s \models tr(\varphi)$ iff $T(M), s \models \varphi$. Thus there is  an adjunction between the maps $tr(.), T$ whose transfer of properties between models has been determined in an abstract setting in \cite{ChuJ}. In fact, $T$ was onto and bijective, which allows for additional transfer of properties between LFD and FO to the extent that LFD can be identified with a fragment of first-order logic. 

The compositional translation $\tau_{\rho}$ from GF into LFD also came with an inverse model transformation G supporting a contravariant equivalence, but this time, it only worked for the special classes of distinguished LFD models and correlated distinguished GF models. However, all models for LFD were LFD-bisimilar to these distinguished models, and an analogous result held for  GF. This relaxed notion of translation up to bisimulation for the relevant two languages seems an interesting generalization of the standard case which still facilitates a good deal of transfer.
We saw this for logical consequence, it also works straightforwardly for decidability, and we even saw how, putting together the translation with its  matching model transformation allowed us to transfer the finite model property.

Finally, our translation from LFD into GF was not entirely compositional, as we also needed to carry special conditions for LFD type models for the formula being translated. This amounts to translating from LFD on its ordinary models into GF on a special class of models satisfying a theory consisting of  special conditions for type models. This theory allowed for a transformation of its GF models into LFD type models, which can then, in a modular fashion, be represented as ordinary LFD dependence models.

Several general questions arise here. One is about the range of the third type of translation technique. We believe that it can deal quite generally with modal-style logics that have an effective `finite quasimodel property' plus a representation theorem.\footnote{As an illustration, consider the modal logic K4 which is complete for transitive models. We cannot translate this directly into GF since transitivity is not guarded. This  can be overcome by suitably translating into the guarded fixed-point logic $\mu$(GF), but a simpler solution is this. Translate the requirements on finite filtrations for K4 models, where in particular, if $\Box\varphi$ belongs to a type, then both $\Box\varphi, \varphi$ belong to accessible types. The result is a guarded description of filtration models.} We leave this matter for further exploration.\footnote{A further interesting issue would be an abstract general transfer analysis of the generalized types of translation that we have provided in this paper.}

\section{Conclusion}

The guarded fragment GF models  restricted quantification, the  modal dependence logic LFD models local dependence between variables. We have presented two new  translation results. One is a faithful compositional translation from GF into the dual CRS fragment of LFD, the other is an effective reduction from  full LFD with dependence atoms to GF. We then demonstrated a number of consequences for transfer of known properties between GF and LFD, and also derived some new results, such as a determination of the  computational complexity of satisfiability in LFD. In summary, local dependence and guarding are much closer as semantic notions than what may have been thought previously. Moreover, the  techniques that we introduce to prove our results may be of wider interest, and we provided some pointers in our final discussion.


\bibliographystyle{aiml22} 
\bibliography{aiml22.bib}

\end{document}